\documentclass[conference]{IEEEtran}
\IEEEoverridecommandlockouts
\usepackage{cite}
\usepackage{amsmath,amssymb,amsfonts}
\usepackage{algorithmic}
\usepackage{graphicx}
\usepackage{textcomp}
\def\BibTeX{{\rm B\kern-.05em{\sc i\kern-.025em b}\kern-.08em
    T\kern-.1667em\lower.7ex\hbox{E}\kern-.125emX}}
\usepackage{mathtools}
\usepackage {bm}
\usepackage{amsthm}
\usepackage[utf8]{inputenc}
\usepackage[english]{babel}
\usepackage{algorithm,algorithmic}
\newtheorem{theorem}{Theorem}
\newtheorem{lemma}[theorem]{Lemma}
\newtheorem{proposition}[theorem]{Proposition}

\newtheorem{definition}[theorem]{Definition}
\DeclareMathOperator*{\argmax}{arg\, max}

\def \Bma {\mathcal{B}}
\def \Bmat {\mathcal{B}_{t}}
\def \Bmath {\mathcal{B}_{t}^{\text{th}}}
\def \Nma {\mathcal{N}}

\def \Gcal {\mathcal{G}}
\def \Ebb {\mathbb{E}}

\def \Rbb {\mathbb{R}}
\def \RbbY {\mathbb{R}^{Y}}
\def \Rbbn {\mathbb{R}^{n}}
\def \Ccal {\mathcal{C}}
\def \Hxy {{\mathcal{H}}_Y^X}

\def \Lxy {{\mathcal{L}}_Y^X}
\def \Mxy {{\mathcal{M}}_Y^X}
\def \Xbfcal {\boldsymbol{\mathcal{X}}}
\def \Xcal {{\mathcal{X}}}
\def \Xmin {{\mathcal{X}}_{\min}}
\def \Xmax {{\mathcal{X}}_{\max}}
\def \hatXmax {\hat{\mathcal{X}}_{\max}}

\def \Xbar {\bar {\mathcal{X}}}
\def \Xmid {{\mathcal{X}}_{m}}
\def \Haxy {\hat {\mathcal{H}}_Y^X}
\def \Laxy {\hat {\mathcal{L}}_Y^X}
\def \Maxy {\hat {\mathcal{M}}_Y^X}
\def \Bhat {\hat {\mathcal{B}}}
\def \xbm {{\bm{x}}}
\def \xbmd {{\bm{x}}^{\prime}}
\def \xbmt {{\bm{x}}_{t}}
\def \xbmtone {{\bm{x}}_{t+1}}

\def \xbmtbar {\bar {x}_{t}}

\def \xbmtybar {{\bar{x}}_{t, -i}}
\def \xti {x_{t,i}}

\def \stysqr {s_{t,-i}^{2}}

\def \sbm {\bm{s}}
\def \sbmt {{\bm{s}}_{t}}
\def \sbmtone {{\bm{s}}_{t+1}}
\def \ut {u_{t}}
\def \uat {u(a_{t})}
\def \ctat {c_{t}(a_{t})}
\def \ctd {c_{t}^{\prime}}
\def \Rcal {\mathcal{R}_{t}}
\def \yt {y_{t}}
\def \ytone {y_{t+1}}
\def \at {a_{t}}
\def \abar {\bar {a}}
\def \atB {a_{t}^{B}}
\def \ats {a_{t}^{*}}
\def \atsn {a_{t}^{*,\text{n}}}
\def \atsp {a_{t}^{*,\text{p}}}
\def \atU {a_{t}^{U}}
\def \atrng {0 \le a_{t} \le {\hat a}}
\def \amax {\hat a}
\def \zt {z_{t}}
\def \ztone {z_{t+1}}
\def \utaz {u_{t}(a_{t};z_{t})}

\def \utc {\pi_{t}}
\def \dutc {\pi^{\prime}_{t}}
\def \utca {\pi_{t}(a_{t})}
\def \utcaz {\pi_{t}(a_{t};z_{t})}

\def \gt {g_{t}}
\def \gT {g_{T}}
\def \qt {q_{t}}
\def \rbp {r}
\def \Vt {V_{t}}
\def \Vtone {V_{t+1}}
\def \Vtn {V_{t}^{\text{n}}}
\def \Vtp {V_{t}^{\text{p}}}

\def \atBH {\{a_{t}^{\text{B}}\}_{t \in \mathbb{H}}}
\def \abar {\bar{a}}

\def \acheck {\check{a}}
\def \ucheck {\check{u}}

\def \Phit {\Phi_{\text{t}}}

\def \ept {\varepsilon_{\text{t}}}
\def \asf {\mathsf{a}}
\def \bsf {\mathsf{b}}
\def \st {s_{t}}
\def \stsqr {s_{t}^{2}}
\def \fbm {{f}}
\def \styt {\sbmt \vert_{\yt = 0}}

\def \thetabm {\boldsymbol{\theta}}
\def \thst {{\theta}^{*}}

\def \atL {a_{t}^{\text{L}}}
\def \aiL {a_{i}^{\text{L}}}
\def \atR {a_{t}^{\text{R}}}
\def \aiR {a_{i}^{\text{R}}}
\def \Dcal {\mathcal{D}}
\def \tht {\theta_t}
\def \thtst {\theta_t^{*}}
\def \thtonest {\theta_{t=1}^{*}}

\begin{document}
\title{Analysis and Evaluation of Baseline Manipulation in Demand Response Programs}

\author{Xiaochu~Wang,~\IEEEmembership{Student Member,~IEEE}~and~Wenyuan~Tang,~\IEEEmembership{Member,~IEEE}
\thanks{Xiaochu Wang and Wenyuan Tang are with the Department of Electrical and Computer Engineering, North Carolina State University, Raleigh, NC 27695 USA e-mail: xwang78@ncsu.edu, wtang8@ncsu.edu.}}

\maketitle

\begin{abstract}
The customer baseline is required to assign rebates to participants in baseline-based demand response (DR) programs. The average baseline method has been widely accepted in practice due to its simplicity and reliability. However, the customer's baseline manipulation is little-known in the literature. We start from a customer's perspective and establish a Markov decision process to model the customer's payoff-maximizing problem. The behavior of a rational customer's underconsumption on DR days and overconsumption on non-DR days are revealed. Furthermore, we propose an approximated baseline method and show how the consumption distribution and program parameters affect the results. Due to the curse of dimensionality, a linear policy-based rollout algorithm is introduced to obtain a practical approximate solution. Finally, a case study is carried out to illustrate the baseline manipulation, where the simulation results confirm the effectiveness of the proposed methods and shed light on how to properly design baseline methods.
\end{abstract}
\section{Introduction}
\IEEEPARstart{D}{mand} response (DR) is considered as an effective and reliable solution to balance the power supply and demand when the power system is under stress \cite{vardakas2015survey}. 
In incentive-based DR programs, the program provider reaches an agreement with customers, who are offered rebates for their participation or active load adjustment according to the real-time instructions \cite{palensky2011demand}. 
The customer baseline calculation is critical to those DR programs under which rebates are paid to the customers for consumption below the baseline. 
If the baseline is overestimated, the program provider overpays the customers which may result in program deficit. Underestimation of baseline may keep customers away from participation. 
Coming up with a perfect baseline method is demanding. 
First, the baseline is the counterfactual consumption that cannot be measured directly.
In addition, the implementation of DR program may motivate the customers to change their consumption patterns, which makes the baseline calculation more difficult \cite{chao2011demand}.

Accuracy, integrity, and simplicity are three critical characteristics of proper baseline methods \cite{rossetto2018measuring}. 
Integrity shows that a good baseline is robust to the customer's attempts to manipulate the system, while the simplicity requires the baseline mechanism to be simple enough for end-use customers. 
The average baseline methods have been widely accepted by system operators and third-party entities, e.g., PJM, NYISO, CAISO, and EnerNOC due to their simplicity and reliability \cite{enernoc2011the}. 
The baseline accuracy analysis has attracted wide attention in the literature. 
The performance of several popular baseline methods is analyzed in \cite{wijaya2014bias} based on the bias metric.
An elaborated error analysis of the current existing baseline methods, including HighXofY, LowXofY, exponential moving average, and regression method, is given in \cite{mohajeryami2017error}. An economic analysis of a peak time rebate program is also carried out. Their case study reveals that the utility pays at least half of its revenue as a rebate solely due to the inaccuracy of baseline methods.

Although the inaccurate estimation of customer baseline is claimed as the main reason for the program inefficiency, baseline manipulation is also a key factor for this issue \cite{wolak2007residential}. Customers are encouraged to inflate their baselines to gain higher rebates if the baseline methods were chosen improperly. Such manipulation in a day-ahead load response program is illustrated in \cite{england2008docket}. Consider the HighXofY method as an example, which picks the highest $X$ data from the last $Y$ days proceeding the DR event and calculates the average of the $X$ data as the customer's baseline. Under this baseline method, customers may intentionally overconsume energy to inflate their baselines, in exchange for a possibly high rebate in the future \cite{wang2018overconsume}. 
For example, customers put all lights on in a stadium during daytime to increase the probability of getting a large amount of rebates \cite{borenstein2014peak}. 

In the literature, there are few mathematical models in quantifying such manipulations.
The customer's manipulation is confirmed as a rational behavior in \cite{vuelvas2015demand}, but their two-stage model is too simple to represent the true scenario. An infinite-horizon model is established to investigate the customer's manipulation behavior \cite{chao2013incentive}. However, their model is only applicable for the HighYofY method with stationary utility function.
The multi-stage stochastic optimization model in \cite{ellman2019model} is also only valid under the HighYofY method.
The incentive-compatible mechanism design is applied to eliminate the baseline manipulation behavior \cite{zhou2017eliciting, li2017mechanism}.
A novel contract is proposed to address the manipulation through self-report baseline in \cite{vuelvas2018limiting}. Nonetheless, the self-report procedure may burden the customers from participation. The oversimplified assumptions cannot address the real-world settings.

The rest of this paper is organized as follows. In Section II, we model the payoff-maximizing problem as a Markov decision process, which is then solved in Section III. 
In Section IV, we propose an approximated baseline method to show more structural results. The approximate dynamic programming is given in Section V. In Section VI, we conduct case studies to illustrate the manipulation behavior. Section VII concludes the paper.

\section{Markov Decision Process Formulation}
A Markov decision process (MDP) is a mathematical framework for solving sequential decision-making problems under uncertainties. For rational customers who want to maximize the expected payoff, they need to consider DR event uncertainty and then identify the optimal policies. Therefore we model the payoff-maximizing problem as a MDP.
\subsection{DR Event}
A DR event indicates the situation with high wholesale market prices or when the system reliability is under stress. The corresponding day is called a DR-event day, or simply, a DR day. A day without a DR event is called a non-DR day. We use $y_t$ to indicate whether day $t$ is a DR day or not, defined as
\begin{equation}
    \label{eq:yt}
    {y_{t}} = \begin{cases}
                1, & \text{if day $t$ is a DR day},\\
                0, & \text{otherwise}.
              \end{cases}
\end{equation}
A DR event typically lasts from one to several hours, which is referred as the DR duration. In this paper, we only consider one-hour duration for convenience.

\subsection{Baseline Method}
The average method, exponential moving average, and regression method are three popular baseline methods in the literature \cite{mohajeryami2017error}.
Under the average method, the baseline is calculated based on the consumption from last $Y$ non-DR days, $\xbmt = (x_{t,1}, x_{t,2},\ldots, x_{t,Y})$. 
We use $\Bma(\xbmt): \RbbY \to \Rbb$ as the function to calculate the baseline under average methods. 
Three typical average methods are HighXofY, LowXofY, and MidXofY, where the baselines are calculated by averaging the highest, lowest, and middle $X$ consumption of $\xbmt$, respectively. Based on the definition, one can formulate the HighXofY and LowXofY methods as\footnote{We only take the $X$ as the variable of the baseline methods since the information $Y$ is included in the recent consumption $\xbmt$.} 
\begin{equation}
\label{eq:Hxy_max}
    \Hxy(\xbmt, X) = (1/X)\mathop {\max }\limits_{1 \le {i_1} \le  \cdots  \le {i_X} \le Y} \{ {x_{t,{i_1}}} +  \cdots  + {x_{t,{i_X}}}\},
\end{equation}
\begin{equation}
    \label{eq:Lxy_max}
    \Lxy(\xbmt, X) = (1/X) \mathop{\min} \limits_{1 \le {i_1} \le  \cdots  \le {i_X} \le Y} \{{x_{t,{i_1}}} +  \cdots  + {x_{t,{i_X}}} \}.
\end{equation}
The MidXofY method takes $X$ and $Y$ with the same parity and calculates the baseline by averaging the $X$ middle consumption of the ordered elements of $\xbmt$ \cite{wijaya2014bias}. Let
$x_{(1)}, x_{(2)}, \ldots, x_{(X)} $ be the middle $X$ consumption (the time index $t$ is removed for convenience), then MidXofY is given by
\begin{equation}
    \label{eq:Mxy}
    \Mxy (\xbmt, X) = \textstyle (1/X) [x_{(1)}+ \ldots, x_{(X)}].
\end{equation}
Among the average methods, HighXofY been widely accepted in practice \cite{wijaya2014bias}:
\begin{itemize}
\item PJM: High4of5 for a weekday, and High2of3 for a weekend DR event.
\item NYISO: High5of10 for a weekday, and High2of3 for a weekend DR event.
\item CAISO: High10of10 for a weekday, and High4of4 for a weekend DR event.
\end{itemize}
As a result, we focus on the baseline analysis for the HighXofY method. The analysis method can also be applied to LowXofY and MidXofY.
\subsection{Customer Utility Function}
A customer obtains benefits from electricity consumption, captured by the utility function $\utaz$, where $\at$ is the customer's electricity consumption and $\zt$ is an external parameter indicating the benefit variation on different days. 
The utility function is defined as non-decreasing and concave in consumption \cite{samadi2010optimal}.
There are several forms of utility functions in the literature, e.g., the quadratic and exponential utility functions. Our model is valid under any non-decreasing and concave utility functions.

Let $\ctat$ be the cost of purchasing electricity from load serving entities. Specifically, we consider the flat retail price $\omega$ and therefore we have $\ctat = \omega \at$. A rational customer will choose the intrinsic baseline $\atB$ as the optimal consumption. The intrinsic baseline can be obtained via
\begin{equation}
    \label{eq:atB}
    \atB \in \argmax_{\atrng} \utcaz,
\end{equation}
where $\utcaz \coloneqq \utaz - \ctat$ is the customer net utility function. The consumption is subjected to a physical constraint $\atrng$.

\subsection{Formulation of the MDP Model}
Consider a finite time horizon of $T$ days, indexed by $t=0,1,\ldots, T-1$, with a terminal stage $T$. The state at day $t$ is given by $\sbmt = (\xbmt, \yt, \zt)$, where $\xbmt = (x_{t,1}, \ldots, x_{t,Y})$ is the recent consumption from the last $Y$ non-DR days (with $x_{t, 1}$ the most recent), $\yt$ is the indicator of DR events, and $\zt$ is an external parameter indicating the variation of the utility function.
The decision variable on day $t$ is the consumption $\at$. The transition from $\xbmt$ to $\xbmtone$ is deterministic, given by 
\begin{equation}
    \label{eq:x_transition}
    \begin{split}
       \xbmtone & = (x_{t+1,1}, \ldots, x_{t+1, Y}) \\ 
       & = \begin{cases}
        (\at, x_{t,1}, \ldots, x_{t,Y-1}), & \text{if } \yt = 0,\\
        \xbmt = (x_{t,1}, \ldots, x_{t,Y}), & \text{if } \yt = 1,
    \end{cases}
    \end{split}
\end{equation}
whereas the transition from $\yt$ to $\ytone$ is captured by the conditional probability $p(\ytone|\yt)$. The immediate payoff of the MDP model is 
\begin{equation}
    \label{eq:immediatePayoff}
    \gt(\sbmt, \at) = \utcaz + \yt \rbp [\Bma(\xbmt) - \at]_{+}, \; \forall\;t<T
\end{equation}
where $[x]_{+} = \max \{x,0\}$, with the terminal payoff as $\gT (\sbm_{T}, a_{T}) = 0$. The net utility function $\utcaz$ is assumed to be concave in $\at$, and the $\rbp$ is the rebate price for load reduction $[\Bma(\xbmt) - \at]_{+}$.

A rational customer's problem is to find an optimal policy $\{\ats(\sbmt),\forall t\}$ to maximize the expected payoff $\Ebb_{\{\yt, \zt, \forall t\}} \left[ \mathop{\sum}_{t=0}^{T}{\gt (\sbmt)} \right]$.
Following a dynamic programming approach, we obtain the Bellman equation as
\begin{equation}
    \label{eq:BellmanEq}
    \Vt(\sbmt) = \mathop{\max}_{\atrng} \{\gt(\sbmt, \at) + \Ebb_{\sbmtone \vert \sbmt}[\Vtone(\sbmtone)]   \},
\end{equation}
for $t<T$, with $V_{T}(\sbm_{T}) = 0$. The maximizers $\{\ats(\sbmt),\forall t\}$ for this optimization problem give the optimal policy.

\section{Optimal Policy and Structural Results}
In this section, we solve the payoff-maximizing problem using dynamic programming. In addition, we establish several structural results, which provide insights into the optimal policies.

\subsection{Properties of the Average Baseline Method}
To derive structural results of the optimal policies, we show the monotonicity, convexity, and supermodularity of the average methods. 
\begin{definition}[Monotonicity]
For any $\xbm = (x_1, \ldots, x_n)$ and $\xbmd = (x_1^{\prime}, \ldots, x_n^{\prime})$ in $\Rbbn$,  $\xbm \le \xbmd$ if $x_i \le x_i^{\prime}, \; \forall {i} = 1, \ldots, n$. A function $f:\Rbbn \to \Rbb$ is \textit{increasing} if for any $\xbm, \xbmd \in \Rbbn$ with $\xbm \le \xbmd$, $f(\xbm) \le f(\xbmd)$. $f$ is \textit{decreasing} if $(-f)$ is increasing.
\end{definition}
\begin{proposition} The average methods, i.e., $\Hxy(\xbmt, X)$, $\Lxy(\xbmt, X)$,
and $\Mxy (\xbmt, X)$, are \textit{increasing} in $\xbmt$. Furthermore, $\Hxy(\xbmt, X)$ is decreasing in $X$, while $\Lxy(\xbmt, X)$ is increasing in $X$.
\end{proposition}

\begin{definition}[Convexity]
A set $\Ccal \subseteq \Rbbn$ is \textit{convex} if for any $\xbm, \xbmd \in \Ccal$ and $\lambda \in [0, 1]$, $\lambda \xbm + (1-\lambda) \xbmd \in \Ccal$. Given a convex set $C \subseteq \Rbbn$, a function $f:\Rbbn \to \Rbb$ is \textit{convex} over set $\Ccal$ if for any $\xbm,\xbmd \in \Ccal$, and $\lambda \in [0,1]$, 
\begin{equation}
    \label{eq:convexity}
    f(\lambda \xbm + (1-\lambda)) \le \lambda f(x) + (1-\lambda) f(\xbmd).
\end{equation}
$f$ is \textit{concave} if $(-f)$ is convex.
\end{definition}
\begin{proposition}
The HighXofY method $\Hxy(\xbmt, X)$ is \textit{convex} in $\xbmt$ since the pointwise maximum operation preserves convexity. The LowXofY method $\Lxy(\xbmt, X)$ is \textit{concave} due to the pointwise minimum operation.
\end{proposition}
\begin{definition}[Supermodularity]
\label{def:modularity}
For any $\xbm = (x_1, \ldots, x_n)$ and $\xbmd = (x_1^{\prime}, \ldots, x_n^{\prime})$ in $\Rbbn$, define their \textit{join} as $\xbm \vee \xbmd = (\max \{x_1, x_1^{\prime}\}, \ldots, \max \{x_n, x_n^{\prime}\})$
and their \textit{meet} as $\xbm \wedge \xbmd = (\min \{x_1, x_1^{\prime}\}, \ldots, \min \{x_n, x_n^{\prime}\})$.
A function $f:\Rbbn \to \Rbb$ is \textit{supermodular} if for any $\xbm, \xbmd \in \Rbbn$,
\begin{equation}
    \label{eq:supermodularity}
    f(\xbm) + f(\xbmd) \le f(\xbm \vee \xbmd) + f(\xbm \wedge \xbmd).
\end{equation}
$f$ is \textit{submodular} if $(-f)$ is supermodular.
\end{definition}
\begin{proposition}
\label{pop:modularity}
The HighXofY method $\Hxy(\xbmt, X)$ is \textit{submodular} in $\xbmt$, and the LowXofY method $\Lxy(\xbmt, X)$ is \textit{supermodular} in $\xbmt$.
\end{proposition}
This proposition can be verified by its definition. 
It is worth noting that the monotonicity of $\Mxy(\xbmt, X)$ in $X$ is underdetermined but relies on the specific values of $\xbmt$.
The convexity and supermodularity of the MidXofY method $\Mxy(\xbmt, X)$ in $\xbmt$ are also underdetermined.
\subsection{Structural Results}
\begin{lemma}
\label{lma:Vt_mono}
The optimal value function $\Vt(\sbmt)$ is increasing in $\xbmt$ under the average method $\Bma(\xbmt)$.
\end{lemma}
\begin{proof}
The rebate received on day $t$ is given by 
\begin{equation}
    \label{eq:rebate}
    \Rcal(\xbmt, \yt, \at) = \yt \rbp [\Bma(\xbmt) - \at]_{+}.
\end{equation}
Since the baseline $\Bma(\xbmt)$ is increasing in $\xbmt$, one can easily show the rebate $\Rcal(\xbmt, \yt, \at)$ is also increasing in $\xbmt$. According to the Bellman equation (\ref{eq:BellmanEq}) and the terminal condition, one can easily show that $\Vt(\sbmt)$ is increasing in $\xbmt$ at stage $T-1$. The conclusion then can be obtained via induction.
\end{proof}
\begin{theorem}
\label{thm:over_under}
Let us define the intrinsic baseline without strategizing as $\atB(\zt) \in \argmax_{\atrng} \utcaz$. Then we have
\begin{equation}
    \label{eq:Over_underconsumption}
    \ats(\xbmt, \yt = 1, \zt) \le \atB(\zt) \le \ats(\xbmt, \yt = 0, \zt),
\end{equation}
where $\ats(\sbmt)$ is the optimal policy with strategizing, given by
\begin{equation}
    \label{eq:ats_optimalPolicy}
    \ats(\sbmt) \in  \argmax_{\atrng} \{\gt(\sbmt, \at) + \Ebb_{\sbmtone \vert \sbmt}[\Vtone(\sbmtone)]   \}.
\end{equation}
\end{theorem}
We investigate a rational customer's strategic behavior by comparing the customer's optimal policies to the intrinsic baseline. The $\ats(\xbmt, 1, \zt)$ and $\ats(\xbmt, 0, \zt)$ represent the optimal policies on DR days and non-DR days, respectively.
The above theorem shows that, to maximize the expected payoff, a customer should underconsume on the DR days and overconsume on the non-DR days. 
Although the DR program helps to alleviate the supply shortage on DR days, it also encourages the customers to consume more on the non-DR days. The analysis of this manipulation level is the main focus in this paper.
The following theorem shows the optimal policy on DR days.
\begin{theorem}
\label{thm:ats_threshold}
The optimal policy on a DR day, $\ats(\xbmt, 1, \zt)$, is a threshold policy in its baseline $\Bma(\xbmt)$, given by
\begin{equation}
    \label{eq:ats_thresholdPolicy}
    \ats(\xbmt, 1, \zt) = \begin{cases}
        \atB, & \text{if } \Bma(\xbmt) < \Bmath,\\
        \atU, & \text{if } \Bma(\xbmt) \ge \Bmath,
    \end{cases}
\end{equation}
where $\atB$ is the customer's intrinsic baseline defined in (\ref{eq:atB}), and $\atU \in \argmax_{\atrng} \{{\utcaz + \rbp(\Bma(\xbmt)-\at)}\}$ is the optimal action with penalty, with $\atU \le \atB$.
The $\Bmath$ is a threshold value that is given by 
\begin{equation}
    \label{eq:Bth_thresholdBaseline}
    \Bmath = \frac{\utc(\atB; \zt)-\utc(\atU; \zt)}{\rbp} + \atU.
\end{equation}
\end{theorem}
The above theorem clearly shows that a rational customer will reduce their consumption only if the calculated baseline reaches a threshold. Otherwise, the customer remains the intrinsic consumption.

The following two lemmas give support to the derivation of structural results for the optimal policy on non-DR days.
\begin{lemma}[Supermodularity composition rule]
\label{lma:compsition_supermodularity}
If a function $f: \Rbb \to \Rbb$ is convex and increasing (decreasing) and another function $g:\Rbbn \to \Rbb$ is \textbf{monotonic} and supermodular (submodular), then $f(g(\xbm))$ in supermodular in $\xbm \in \Rbbn$.
\end{lemma}
The proof of this lemma can be found in \cite{simchi2005logic}. 
The following lemma is helpful in establishing comparative static results when the objective function is not differentiable \cite{topkis1978minimizing}.
\begin{lemma}[Topkis's theorem]
\label{lma:Topkis}
If $f$ is supermodular in $(\xbm, \vartheta)$, and $\Gcal$ is a lattice, then $\xbm^{*}(\vartheta) \in \argmax_{\xbm \in \Gcal} {f(\xbm, \vartheta)}$
is increasing in $\vartheta$. If $f$ is submodular in $(\xbm, \vartheta)$, then $\xbm(\vartheta)$ is decreasing in $\vartheta$. 
\end{lemma}
 
\begin{theorem}
\label{thm:ats_nonDR_supermodular}
The optimal policy on a non-DR day, $\ats(\xbmt, 0, \zt)$, is increasing in $\xbmt$ if the average baseline $\Bma(\xbmt)$ is supermodular in $\xbmt$.
\end{theorem}
According to proposition \ref{pop:modularity}, only the LowXofY baseline method $\Lxy(\xbmt, X)$ (includes HighYofY) is supermodular in $\xbmt$. This implies that customers tend to consume more on non-DR days with higher recent consumption. 
The characteristics of HighXofY method and the nonnegative rebate form fail us to have a concave optimization problem, which brings difficulties in establishing explicit structural results. This motivates us to formulate approximate model and method to establish more structural results or to obtain some numerical solution.

\section{The Approximated Baseline Method}
The impact from program parameters, i.e., $X$ and $Y$, on the optimal policies is little-known under the average method. In addition, how the distribution of the customer's consumption affects the calculated baseline is still unclear. 
To address these issues, we propose an approximated baseline method that takes the program parameters explicitly under which more structural results are revealed.

\subsection{Approximation of Baseline Method}
We use the sample mean and extreme order values to formulate an approximation of the baseline. For convenience, we remove the index $t$ and let $\Xbfcal = (\Xcal_{1}, \ldots, \Xcal_{Y})$ be the recent consumption vector consisting of $Y$ independent random variables. To define the approximation of HighXofY, we introduce the largest order statistic (or sample maximum) $\Xmax = \max \{\Xcal_{1}, \ldots, \Xcal_{Y}\}$ and
the sample mean $\Xbar = (1/X) (\Xcal_{1} + \cdots + \Xcal_{Y})$. 
The approximated HighXofY is essentially the convex combination of the sample mean $\Xbar$ and sample maximum $\Xmax$, given by
\begin{equation}
    \label{eq:Hxy_appr}
    \Haxy(\Xbfcal, X, Y) = \lambda \Xbar + (1-\lambda) \Xmax,
\end{equation}
where $\lambda= \frac{X-1}{Y-1} \in [0,1]$ is the weight of the sample mean. 
Clearly, when $\lambda = 0\;(X=1)$ we have $\Haxy(\Xcal, X, Y) = \Xmax$, and $\Haxy(\Xcal, X, Y) = \Xbar$ while $\lambda = 1\;(X=Y)$. This also shows $\Haxy (\Xcal, X, Y) = \Hxy (\Xcal, X, Y)$ under these two cases.
Similarly, after introducing the sample minimum $\Xmin = \min\{\Xcal_{1}, \ldots, \Xcal_{Y}\}$ and sample median $\Xmid$, one can define the approximated LowXofY and MidXofY as $\Laxy (\Xbfcal, X, Y) = \lambda \Xbar + (1-\lambda) \Xmin$ and $\Maxy (\Xbfcal, X, Y) = \lambda \Xbar + (1-\lambda) \Xmid$, respectively.

The $\Xmax$ in (\ref{eq:Hxy_appr}), however, prevents us deriving structural results since $\Xmax$ is essentially the High1ofY method. Therefore, we come up with an approximation for the sample maximum $\Xmax$.
Reference \cite{chen1999accurate} gives an effective formula to estimate the sample maximum from the normal distribution. 
Suppose the standard normal distribution $\Nma(0,1)$ is truncated in the interval $[\alpha, \beta]$, denoted by $\Nma(0,1,\alpha, \beta)$. Then we have the estimation of sample maximum for $\Nma(0,1,\alpha, \beta)$ as
\begin{equation}
    \label{eq:xmax_appr_trun}
    \fbm(Y) = \Phit^{-1}[(0.5+\ept)^{1/Y}],
\end{equation}
where $\Phit{(\cdot)}$ is the cumulative distribution function (CDF) of the truncated normal distribution, $\ept$ is a parameter calculated by a bisectional searching algorithm that optimizes the match of (\ref{eq:xmax_appr_trun}) to the expected values \cite{chen1999accurate}. 
The customer consumption is assumed to obey a truncated normal distribution $\Nma (\mu, \sigma, \asf, \bsf)$, which is converted into the corresponding standard normal distribution with truncation $\Nma (0,1,\alpha, \beta)$, with $\alpha = \frac{\asf-\mu}{\sigma}$ and $\beta = \frac{\bsf-\mu}{\sigma}$.
Then the estimation of the sample maximum is $\hatXmax \approx  \mu + \fbm(Y) \sigma$.
For a specific sample $\xbmt$, one can estimate the sample maximum $x_{t,\max} = \max_{1 \le i \le Y} \{x_{t,1}, \ldots, x_{t,Y}\}$ using
\begin{equation}
    \label{eq:xmax_sample_mean_variance}
    \hat {x}_{t,\max} = \xbmtbar + \fbm(Y) \st,
\end{equation}
where $\xbmtbar = \frac{1}{Y}(x_{t,1}+ \cdots + x_{t,Y})$ is the sample mean and $\st = \sqrt{\stsqr}$ is the standard deviation with $\stsqr = \frac{1}{Y-1} \mathop{\sum}_{i=1}^{Y} {(\xti - \xbmtbar)^{2}}$. 
Then, based on (\ref{eq:Hxy_appr}), one can obtain the final approximated HighXofY in terms of the recent consumption $\xbmt$, $X$, and $Y$ as 
\begin{equation}
    \label{eq:Hxy_final}
    \Haxy(\xbmt, X, Y) = \xbmtbar + \frac{Y-X}{Y-1} \fbm(Y) \st,\; \forall\;Y \ge 2.
\end{equation}

To quantify the goodness-of-fit of this approximation, we define the percent error of approximation as $(\Bma - \Bhat)/{\Bhat} \times 100\%$,
where $\Bhat$ is the approximated baseline. The $\Bhat$ is used as the denominator because it provides a more stable reference for illustration than the actual baseline. To show the feasibility of this approximation, we use both real data and synthetic data (with additive white Gaussian noise, which can be found in the case study), where the results are shown in Fig. \ref{fig:app_err}. 
Although there are some outliers, the absolute approximation errors are less than 5\%. 
These results show that the proposed approximated baseline method is effective.
\begin{figure}[!t]
    \centering
    \includegraphics[width=0.95\linewidth]{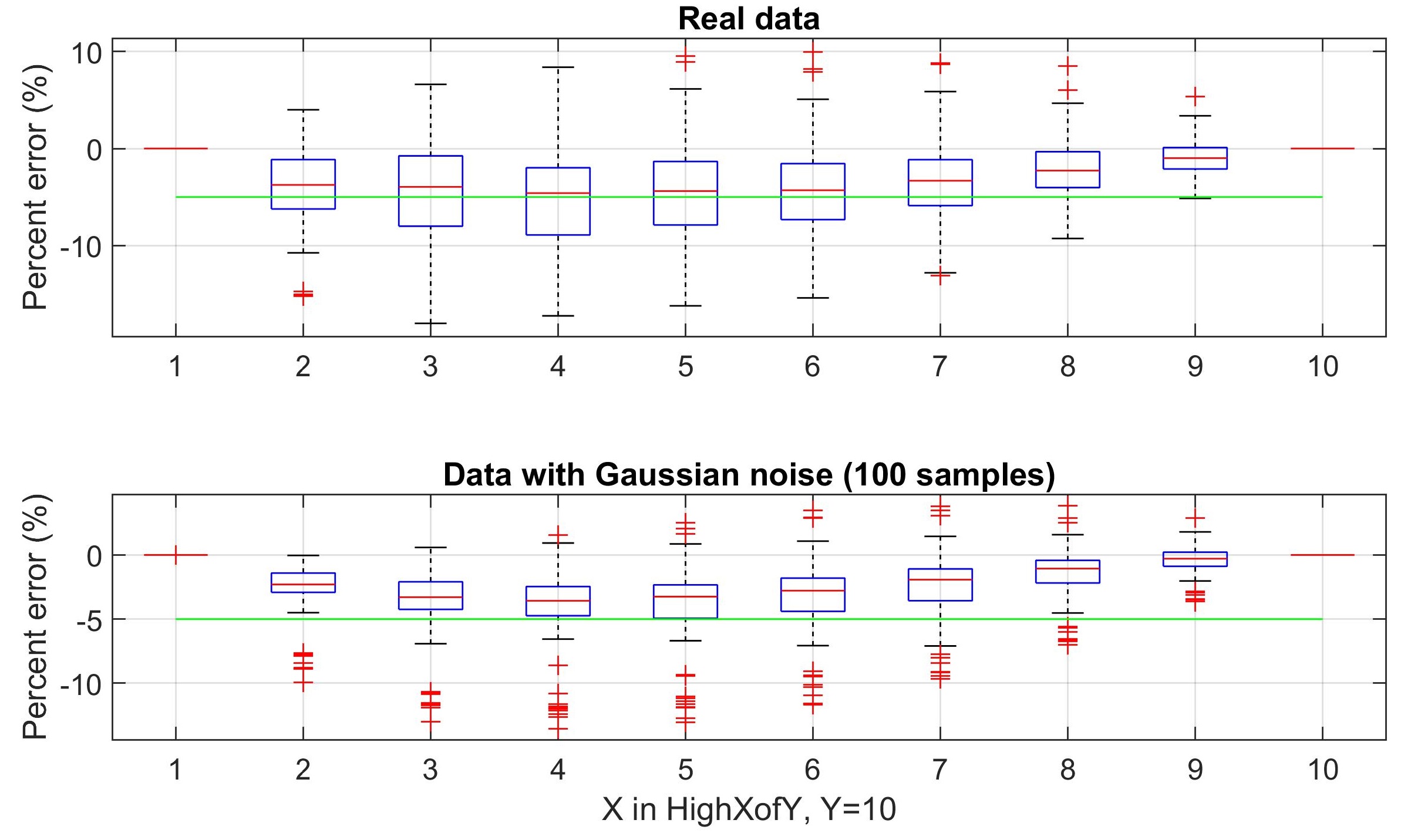}
    \caption{The approximation error for the approximated HighXofY methods (the green line indicates the -5\% error).}
    \label{fig:app_err}
\vspace{-3mm}
\end{figure}

\subsection{Structural Results under the Approximated Baseline}
\subsubsection{The monotonicity of the $\Haxy(\xbmt, X, Y)$} We first show the monotonicity of the approximated baseline in its parameters in the following proposition.
\begin{proposition}
\label{prop:Haxy_monotonicity}
The $\Haxy(\xbmt, X, Y)$ is increasing in $\xbmt$, but decreasing in $X$. The expectation $\Ebb[\Haxy]$ is increasing in $\mu$ and $\sigma$. The monotonicity of $\Haxy(\xbmt, X, Y)$ in $Y$ is underdetermined but relies on the selection of $X$.
\end{proposition}
For the monotonicity of $\Haxy$ in $Y$, we give the analysis of four selections of $X$. 
(i) $X=1:$ this shows that $\Haxy = \xbmtbar + \fbm(Y) \st$ is increasing in $Y$. 
(ii) When $\frac{Y-X}{Y-1} \approx \xi,\;\xi \in (0,1)$ (we can take $X = \left \lfloor (1-\xi)Y + \xi \right \rfloor$), then $\Haxy \approx \xbmtbar + \xi \fbm(Y)\st$ is increasing in $Y$. 
(iii) When $X=Y-K$, $K\in \mathbb{Z}^{+},\;K<Y$, then $\Haxy = \xbmtbar + \frac{K}{Y-1} \fbm(Y)$ is decreasing in $Y$ since $\fbm(Y)$ is sublinear in $Y$.
(iv) When $X=Y$, $\Haxy = \xbmtbar$ is independent of $Y$. 
The monotonicity of $\Haxy$ in $Y$ can be illustrated with a simple example. The calculated baselines $\{\hat {\mathcal{H}}_4^2$, $\hat {\mathcal{H}}_6^3$, $\hat {\mathcal{H}}_8^4$, $\hat {\mathcal{H}}_{10}^5\}$ are increasing in $Y$, while $\{\hat {\mathcal{H}}_3^2$, $\hat {\mathcal{H}}_5^4$, $\hat {\mathcal{H}}_8^7$, $\hat {\mathcal{H}}_{10}^9\}$ are decreasing in $Y$. 

\subsubsection{The customer's manipulation}
\begin{proposition}
The approximated HighXofY, $\Haxy(\xbmt)$, is submodular in $\xbmt$.
\end{proposition}
\begin{proof}
We have the $\Haxy(\xbmt, X, Y) = \xbmtbar + \frac{Y-X}{Y-1}\fbm(Y)\st$, where the standard deviation can be written as
\begin{equation}
    \label{eq:st_standard_deviation}
    \st(\xbmt) = \sqrt{\frac{1}{Y-1} \left( \mathop{\sum}_{i=1}^{Y}{\xti^{2}} - Y \xbmtbar ^{2} \right)},
\end{equation}
with $\xbmtbar = \frac{1}{Y}(x_{t,1}+ \cdots + x_{t,Y})$. The submodularity of $\st(\xbmt)$ can be verified via the composition rule from lemma \ref{lma:compsition_supermodularity}. 
\end{proof} 
The above proposition shows that the $\Haxy$ has the same submodularity with $\Hxy$.
However, we can apply the same approach used in Theorem \ref{thm:ats_nonDR_supermodular} to establish the monotonicity of $\ats(\styt)$ in the program parameter $X$ or $Y$. On a non-DR day, the optimal policy is obtained by
\begin{equation}
    \label{eq:ats_nonDR_optimal_policy}
    \ats(\xbmt, 0, \zt) \in \argmax_{\atrng} \{ \utca + \Ebb_{\sbmtone \vert \yt = 0}[\Vtone(\sbmtone)]\},
\end{equation}
with $\xbmtone = (\at, x_{t,1}, \ldots, x_{t, Y-1})$. To show the monotonicity of $\ats(\styt)$ in $X$ or $Y$, one only needs to show that the expectation term is supermodular in $(\at, X)$ or $(\at, Y)$,
i.e., to show $\Vt(\xbmt, X, Y)$ is supermodular in $(\xti, X)$ or $(\xti, Y)$ for any $i = 1, \ldots, Y$. Based on the result from Theorem \ref{thm:ats_nonDR_supermodular}, this is equivalent to show the supermodularity of $\Haxy(\xbmt, X, Y)$.
\begin{theorem}
\label{thm:monotonicity_X}
Given a fixed value $Y$, the optimal policy on a non-DR day $\ats(\styt)$ is increasing in $X$ when the intrinsic consumption on that day is lower than the population average, and decreasing in $X$ otherwise.
\end{theorem}
The above theorem can be illustrated by the following examples.
For instance, on a low-consumption day, a rational customer tends to inflate more under the High10of10 compared to that under the High1of10. This is reasonable because everyday matters under the High10of10 but the customer has other days to inflate under the High1of10. 
On the other hand, on a high-consumption day, a rational customer is prone to less inflation with the increase of $X$. In other words, a small $X$ makes the customer inflate more on the high-consumption days. For example, under the High1of10, a customer tends to increase the load to secure a higher baseline for the future.

\section{Approximate Dynamic Programming via Rollout Algorithm}
Solving the MDP problem using dynamic programming is computationally challenging due to the curse of dimensionality \cite{powell2007approximate}. Therefore, approximate methods are required to solve the MDP problem with a large state space \cite{bertsekas1999rollout}.
In this section, we apply the rollout algorithm to obtain a simulation-based forward dynamic programming solution. 
\subsubsection{The rollout algorithm}
The key idea of the rollout algorithm is starting from some given heuristic and constructing another policy with better performance \cite{bertsekas2013rollout}. 
Essentially, the rollout algorithm is an online forward dynamic programming procedure that selects actions to obtain the maximum expected payoff calculated based on the given heuristic policy. 
The improved performance is guaranteed by the performance improvement properties \cite{goodson2017rollout}, which we show later.
Rollout algorithms step forward in time, which requires estimating the payoff-to-go when evaluating decision rules. 
Specifically, given a heuristic policy, we calculate the expected payoff-to-go under the heuristic based on a Monte Carol simulation \cite{bertsekas1996neuro}. The formula in finding the optimal actions under the rollout algorithm is given as
\begin{equation}
    \label{eq:atR}
    \atR \in \argmax_{\atrng} \left\{\gt(\st, \at) + \displaystyle \Ebb_{\{s_i\}} \left[ \mathop{\sum}_{i=t+1}^{T} g_i(s_i, \aiL(s_i)) | \st  \right] \right\},
\end{equation} 
where $\atL(\st)$ is the given heuristic policy. The sample paths in calculating the expected payoff-to-go can be generated by applying additive white Gaussian noise to customer historical consumption. The DR events along the path are generated according to the conditional probability $p(\ytone|\yt)$. The details of generating sample paths are given in the case study.

\subsubsection{The heuristic policy}
A heuristic policy is any method to select decision rules in each stage within the space of decision rules. In this paper, a linear policy is selected as the heuristic for the rollout algorithm. A linear policy with $k$ parameters can be $a(s; \thetabm) = \theta_1 \cdot \phi_1 (s) + \cdots + \theta_k \cdot \phi_k (s)$, where $\phi_1 (s), \ldots, \phi_k (s)$ are the extracted features and $\thetabm =(\theta_1, \ldots, \theta_k)$ is the corresponding parameter vector. 
Given the optimal policy on DR days in (\ref{eq:ats_thresholdPolicy}), we approximate the policy on the non-DR days based on one feature, given by
\begin{equation}
    \label{eq:LinearPolicy_r1}
    \atL(\sbmt \vert_{\yt = 0}) = \atB + \tht \phi(\sbmt),
\end{equation}
where the feature $\phi(\sbmt)$ can be the maximum of the recent consumption. The idea of this linear policy is to facilitate customers with an easy way to inflate their baseline based on the recent consumption. The best parameter $\thtst$ can be obtained by maximizing the expected payoff from the current stage $t$ to the end stage $T$, i.e.,
\begin{equation}
    \label{eq:thtst}
    \thtst \in \argmax_{0 \le \tht \le 1} \Ebb_{\{s_i\}} \left[ \mathop{\sum}_{i=t}^{T} g_i(s_i, \aiL (s_i)) | \st  \right],
\end{equation}
where the linear policy $\atL(\st)$ on non-DR days and DR days are given in (\ref{eq:LinearPolicy_r1}) and (\ref{eq:ats_thresholdPolicy}). The sample paths for calculating the expected payoff-to-go are also generated by apply additive white Gaussian noise. For simplicity, a stationary linear policy $\thst = \thtonest$ is applied for simulation.

\subsubsection{Performance improvement property \cite{goodson2017rollout}}
The performance improvement property of the rollout algorithm is illustrated in the following definition.
\begin{definition}[Rollout Improvement Property]For a heuristic $\atL$ and a rollout policy $\atR$, we say $\atR$ rollout improving if for $t = 1, 2, \ldots, T$,
\begin{equation}
    \label{eq:Improvement}
    \Ebb \left[ \mathop{\sum}_{i=t}^{T} g_i(s_i, \aiL (s_i)) | \st  \right] \le 
    \Ebb \left[ \mathop{\sum}_{i=t}^{T} g_i(s_i, \aiR (s_i)) | \st  \right].
\end{equation}
\end{definition}
This rollout improvement property is assured based on the construction of $\atR$ in (\ref{eq:atR}) and $\atL$ in (\ref{eq:thtst}).

\section{Case Study}
In this section, we conduct a case study to show the customer's manipulation behavior under the HighXofY methods. Customer utility function is estimated from historical data. Due to the curse of dimensionality, we only give the results using dynamic programming with $Y \le 7$. Cases with larger $Y$s are conducted via the rollout algorithm.

\subsection{Bias and Data Set}
We introduce the metric of bias to evaluate the performance of the baseline method. Let $\atBH$ be the customer's historical consumption at the DR-event hour and $\Dcal$ be the set of the DR days. Then we can define the bias as
\begin{equation}
    \label{eq:bias_DR}
    \frac{\mathop{\sum}_{t\in \Dcal} {[\Bmat(\xbmt) - \atB}]}{\mathop{\sum}_{t\in \Dcal} {\atB}} \times 100\%,
\end{equation}
where $\atB$ is the intrinsic consumption and $\Bma(\xbmt)$ is the calculated baseline. 
By comparing the biases from the base case (without DR programs) with the manipulation case (with DR programs), one can identify the customer baseline manipulation level.

To properly define the reference bias, we first introduce the customer's historical consumption $\atBH$ that comes from the \textit{UMass} smart home data set \cite{barker2012smart}. This data set consists of load and weather data for 114 households for every 15-minute interval for two years and three months. Here, one data set from a typical customer\footnote{We choose the customers based on two criteria: (i) without bad data and (ii) showing a morning peak.} (Apt \#21) in 2015 is used for illustration. We use the hourly consumption on the weekdays and remove the data on holidays.
We set 9:00 am as the DR hour (with the highest average consumption). Furthermore, we consider a total of 93 days from peak consumption seasons.
However, the limited data set cannot serve a stable reference in evaluating baseline manipulation. To address this issue, we apply additive white Gaussian noise to the initial data set and create large data samples.
As a result, the time-series information of the initial data set can be preserved and the generated large data set can provide a stable reference. The initial data and the data with additive white Gaussian noise are illustrated in Fig. \ref{fig:data_Gaussian}. The signal-to-noise ratio is set as 3dB to balance the time-series information and the consumption randomness. Note that the lower bound is set to zero to avoid negative consumption.
\begin{figure}[!t]
    \centering
    \includegraphics[width=.9\linewidth]{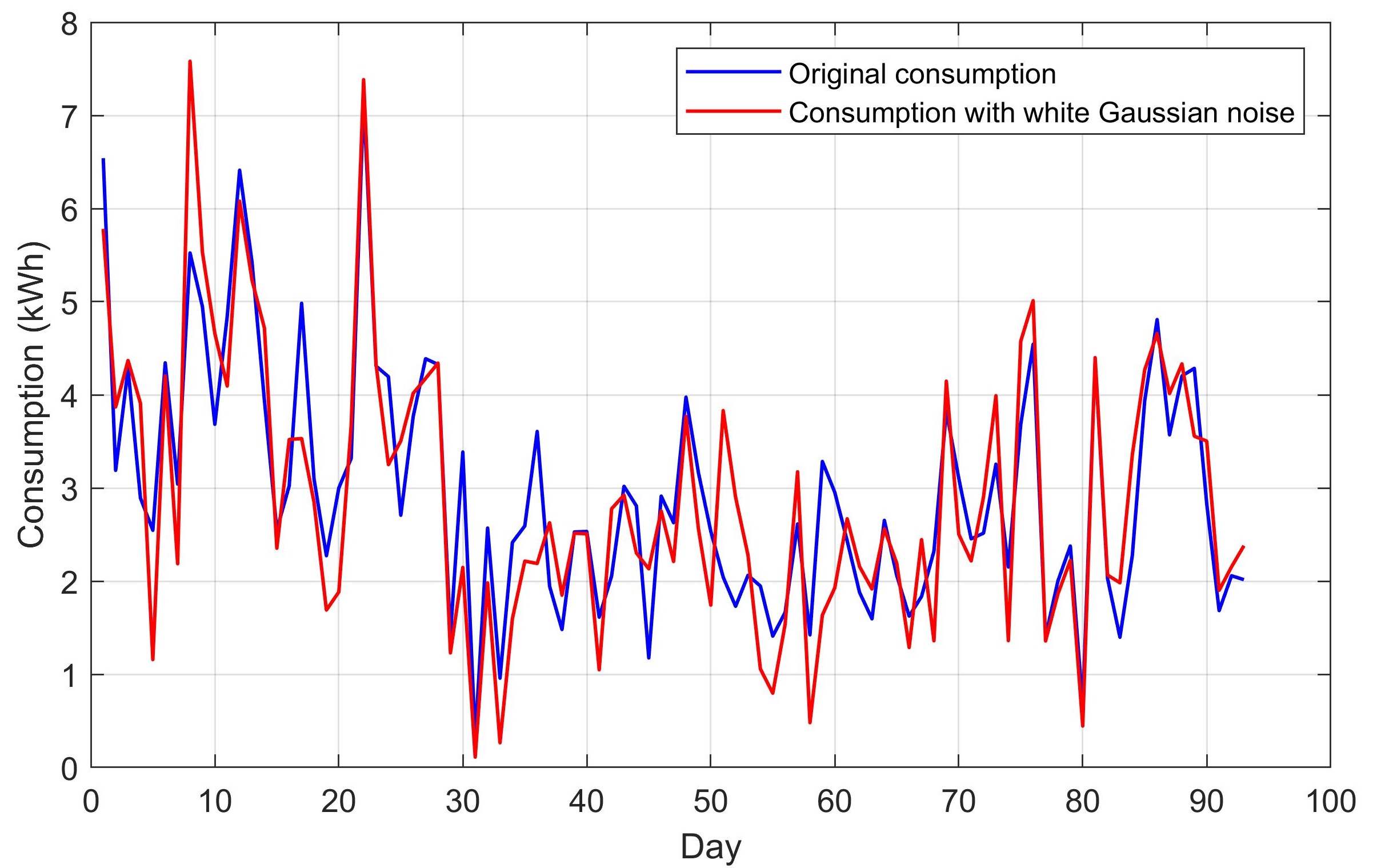}
    \caption{Illustration of the initial data and the data with additive white Gaussian noise.}
    \label{fig:data_Gaussian}
\vspace{-3mm}
\end{figure}

In addition, we generate independent DR event sequences to increase the randomness of the sample path, which is a $T=93$ days consumption and DR signals.
A total of $N = 100$ paths are generated for simulation where the DR events are initiated according to the conditional probability $p(\ytone|\yt)$.
An overall around 30\% DR rate is considered in this paper. Considering the fact that the DR event is more likely activated when the previous day was a DR day, we set the conditional probabilities as $p_0 \coloneqq p(\ytone = 1|\yt = 0) = 0.2$ and $p_1 \coloneqq (\ytone = 1|\yt = 1) = 0.4$.

\subsection{Estimation of Customer Utility Function}
Given a customer's historical consumption $\atBH$, we propose an innovative way to construct the customer utility function. The consumption mean $\abar$ and upper bound $\acheck$ of the historical consumption are used for this construction.
We define the utility function as $\utaz = \zt \cdot \uat$, where $\uat$ is the reference utility under which the consumption mean $\abar$ is the customer's optimal decision. Besides, we use the exponential form of the utility function as $\uat = \gamma (1-e^{-\frac{\at}{\rho}})$,
where $\gamma$ and $\rho$ are two parameters called utility ratio and utility shape, respectively. 
The determination of these two parameters is given below.

\textit{The utility shape $\rho$:}
we introduce the maximum relative utility as $\ucheck \coloneqq \frac{\ut(\acheck, \zt)}{ \lim_{\at \to \infty} {\utaz}} = 1-e^{-\frac{\acheck}{\rho}}$.
In practice, we choose $\ucheck = 0.99$\footnote{This is a subjective choice and any value that is close to 1 is reasonable. Furthermore, the results are not very sensitive to the selection of $\ucheck$.}. Therefore, the utility shape parameter $\rho$ can be determined as
\begin{equation}
    \label{eq:rho}
    \rho = -\frac{\acheck}{\ln(1-\ucheck)}.
\end{equation}

\textit{The utility ratio $\gamma$:}
$\gamma$ is defined as the ratio of one-unit utility to one monetary value (\$). According to the definition of $\uat$, i.e., under which the average consumption $\abar$ is optimal, we have $\abar \in \argmax_{\atrng} [\uat - \ctat]$.
Based on the first-order optimality condition, the derivative is zero at $\abar$, i.e., $\label{eq:muHequation}
    u^{\prime}(\abar) \coloneqq (\gamma/\rho) e^{-\abar/\rho} = \ctd (\abar) = \omega$.
Therefore, we have the utility ratio as
\begin{equation}
    \label{eq:gamma}
    \gamma = \omega \rho e^{\frac{\abar}{\rho}}.
\end{equation}

Based on the historical consumption data from the customer \#21, we have the utility parameters as $\rho = 1.56$ and $\gamma = 1.25$. Fig. \ref{fig:utilities} illustrates the customer utilities under $\zt = [0.8, 1, 1.2]$, shown as the blue curves. The net utility function, shown as the red curve, is concave in consumption.
\begin{figure}[!t]
    \centering
    \includegraphics[width=.9\linewidth]{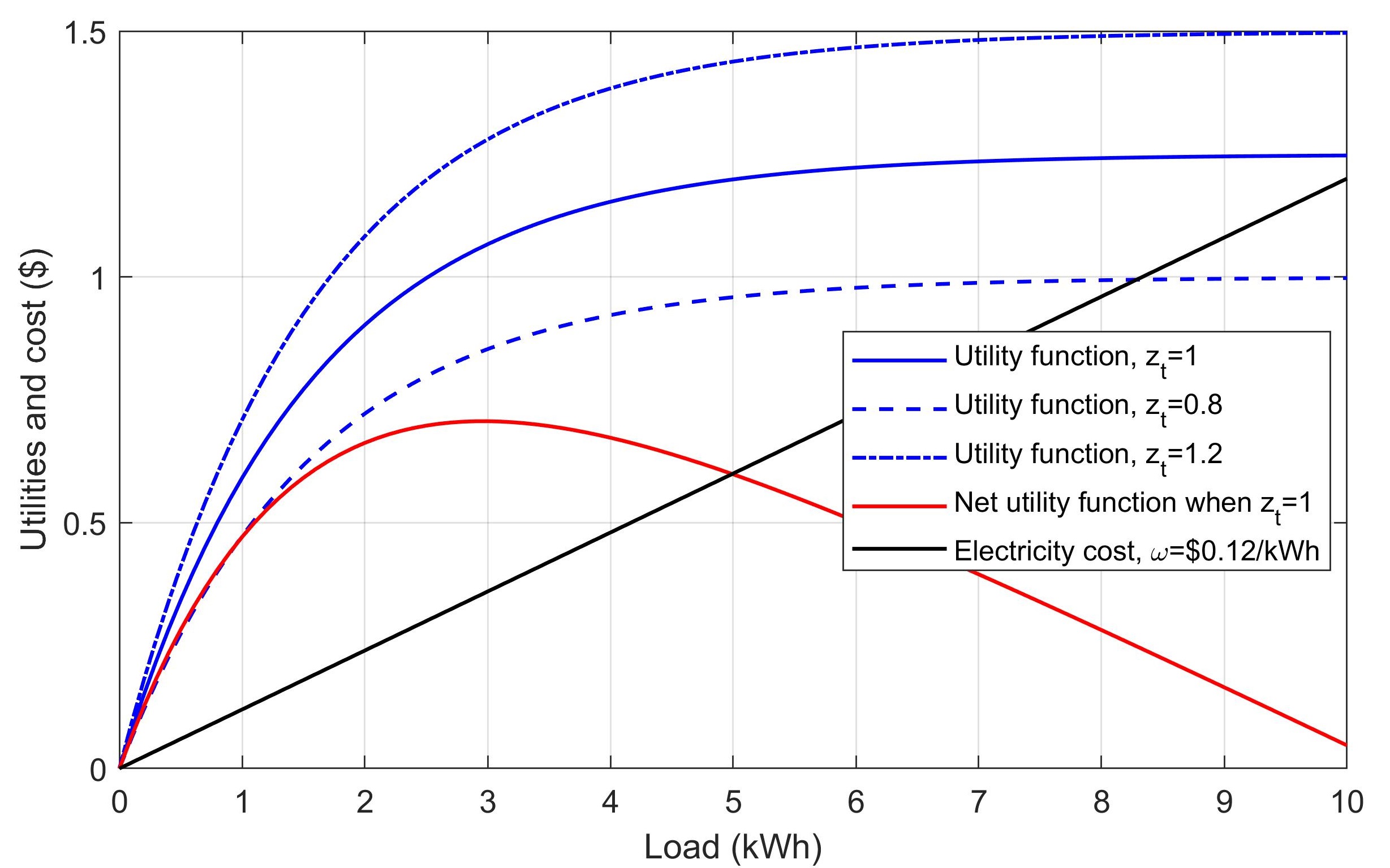}
    \caption{The customer's utility functions and net utility function.}
    \label{fig:utilities}
\vspace{-3mm}
\end{figure}

\subsection{Results via Dynamic Programming}
We solve the MDP problem by following the dynamic programming approach shown in (\ref{eq:BellmanEq}). 
To find the optimal policies, we discretize the action into 10 discrete values. 
Due to the curse of dimensionality, we only consider $Y \le 7$ here. 
\subsubsection{Bias and manipulation in terms of $X$ in HighXofY}
Fig. \ref{fig:bias_mani_Y5_U_X} shows the results for HighXofY when $Y=7$. 
We first observe that the bias is decreasing in $X$, which matches the monotonicity of HighXofY in $X$.
The bias is around zero when $X=7$ because the DR events are independently generated of customer consumption.
The bias with DR manipulation is much higher than the counterparts without DR. 
The difference between the bias with and without DR is captured by the manipulation curves (red curves). 
We also notice that the manipulation is a single-peaked function in $X$, where $X=$ 6 and 7 leads to the lowest manipulation level.
\begin{figure}[!t]
    \centering
    \includegraphics[width=.9\linewidth]{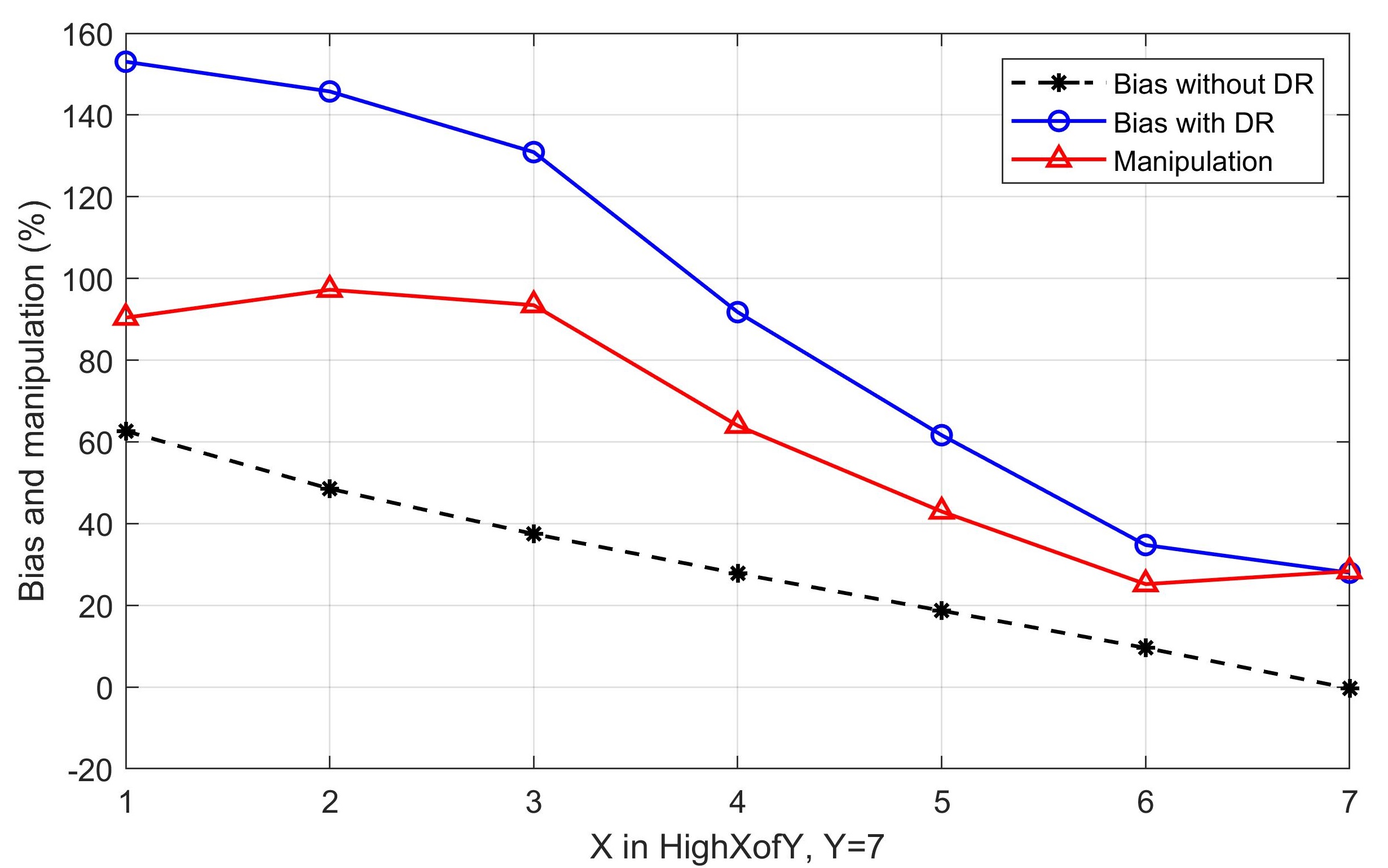}
    \caption{The bias and manipulation in terms of $X$ when $Y=7$ and $\rbp = \$0.12$/kWh.}
    \label{fig:bias_mani_Y5_U_X}
\vspace{-4mm}
\end{figure}
\subsubsection{Results in terms of rebate price $\rbp$} 
We compare the results under three rebate prices, i.e., \$0.06/kWh, \$0.12/kWh, and \$0.18/kWh, with $\omega$ = \$0.12/kWh. The curves in Fig.~\ref{fig:bias_mani_Y5_pir_X} show that the higher of the rebate price, the more of the customer manipulation. The three curves also preserve the single-peaked property. 
When the rebate price is \$0.06/kWh, the manipulation is relatively low when $ X\ge 3$ and almost zero when $X=7$. This shows that a low rebate price and a large $X$ can effectively alleviate the customer's manipulation. 
\begin{figure}[!t]
    \centering
    \includegraphics[width=.9\linewidth]{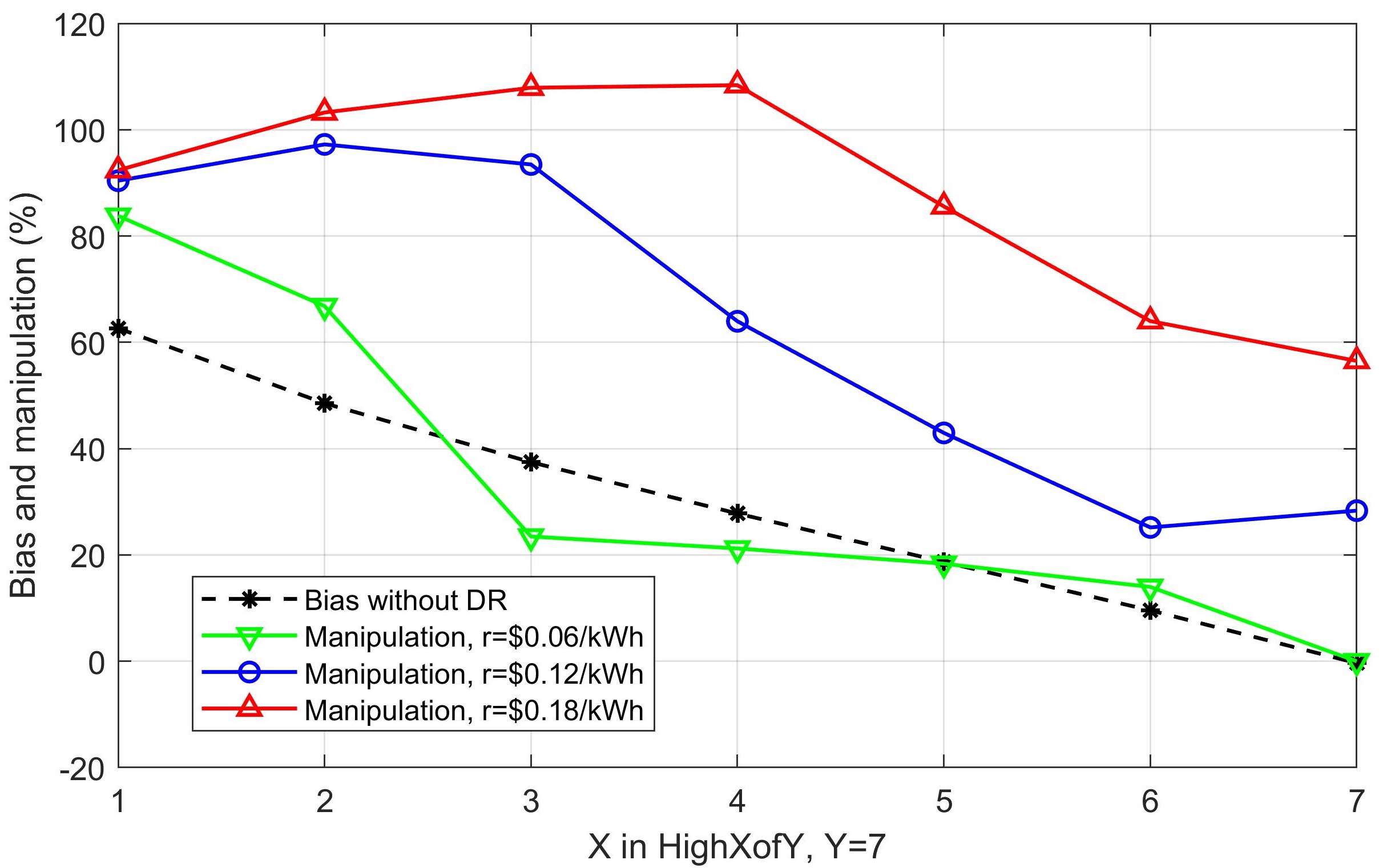}
    \caption{The bias and manipulation in terms of $\rbp$ and $X$ when $Y=7$.}
    \label{fig:bias_mani_Y5_pir_X}
\vspace{-4mm}
\end{figure}
\subsubsection{Results under HighYofY method}
Fig. \ref{fig:HighXofY357} shows the manipulation under the HighXofY method when $Y=3,\ldots,7$. 
All the manipulation curves preserve the single-peaked property and the shapes look like a tilted and stretched letter \textbf{Z} when $Y \ge 5$. 
The manipulation level remains almost the same under the HighYofY method.
\begin{figure}[!t]
    \centering
    \includegraphics[width=.9\linewidth]{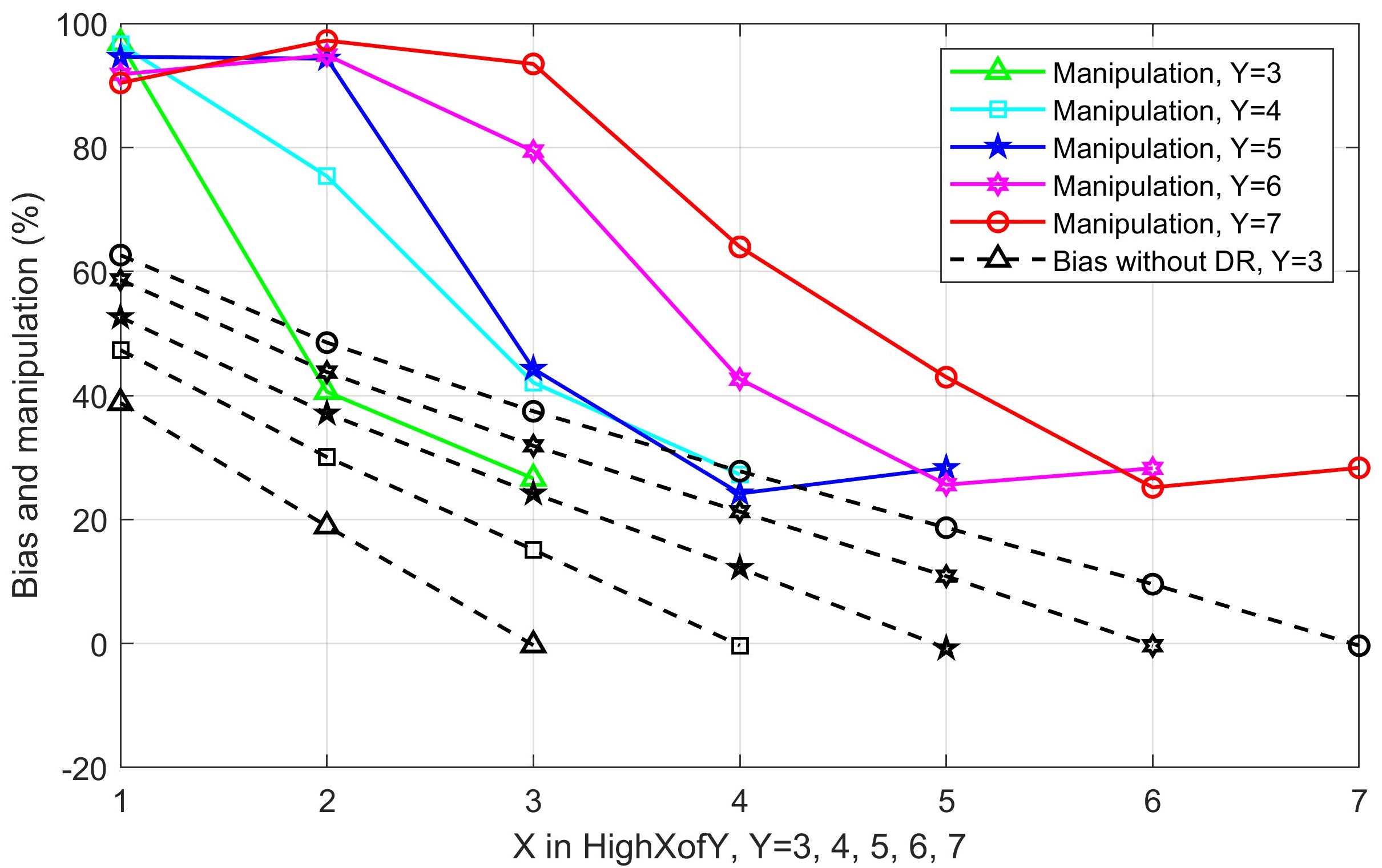}
    \caption{The bias and manipulation for HighXofY, $Y=3,\ldots,7$, $\rbp = \$0.12$/kWh.}
    \label{fig:HighXofY357}
\vspace{-4mm}
\end{figure}
\subsection{Results via Rollout Algorithm}
We generate 10 paths to calculate the expected payoff-to-go for the linear policy, where the best parameter is
$$ 
\theta^{*}(X) = \begin{cases}
                0.148, & \text{when } X=1,\ldots, 7,\\
                0.125, & \text{when } X=8, 9, 10.
              \end{cases}
$$
When Y = 7, we have the same threshold parameters, with $\theta^{*}(X) = 0.148$ for $X=1,\ldots, 5$, and $\theta^{*}(X) = 0.125$ for $X= 6,7$.
Another 1000 paths are generated to obtain the rollout policy (\ref{eq:atR}).
Fig. \ref{fig:exact_vs_rollout} shows the simulation results for HighXofY when $Y=$ 7 and 10. The two blue curves are the manipulation levels under the rollout algorithm while the red curve shows the exact result for $Y=7$. The two approximated curves preserve the single-peaked property. The manipulation levels drop quickly around $X = Y/2$ and remain almost the same when $X$ is close to $Y$. 
We also compare the results from dynamic programming and from the rollout algorithm when $Y=7$: although there are some differences in manipulation when $X \le 5$, these two curves are similar in shape and almost match when $X=$ 6 and 7. This comparison confirms the effectiveness of the rollout algorithm.
\begin{figure}[!t]
    \centering
    \includegraphics[width=0.95\linewidth]{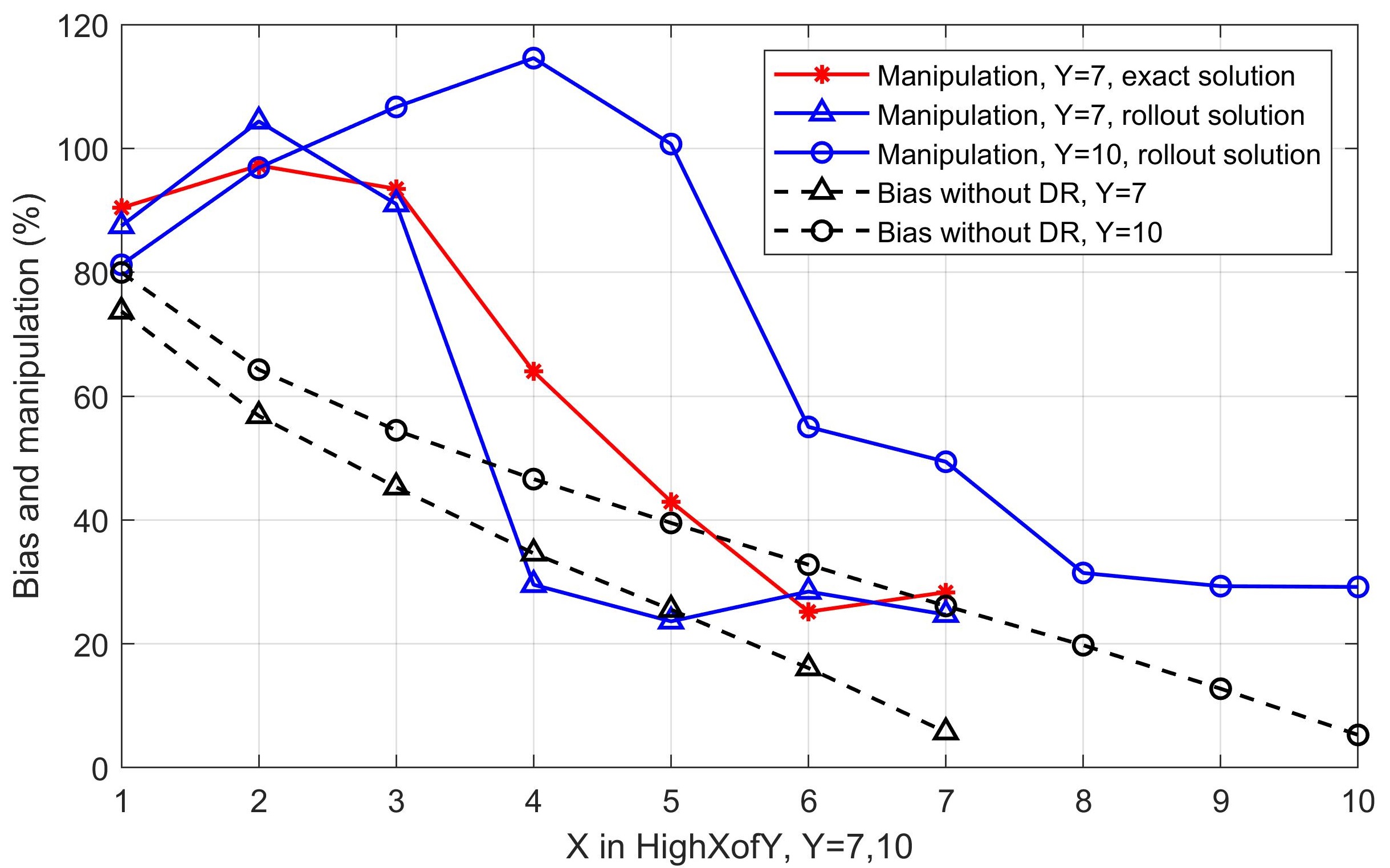}
    \caption{Comparison between the exact solution and the rollout solution for HighXofY, $Y=$ 7 and 10, $\rbp = \$0.12$/kWh.}
    \label{fig:exact_vs_rollout}
\vspace{-4mm}
\end{figure}
 
\section{Conclusion}
In this paper, we give the baseline manipulation analysis of the used accepted average baseline methods. 
We establish a Markov decision process to formulate the customer's profit-maximization problem. By comparing the optimal policies to the intrinsic baseline, we verify a rational customer's behavior: underconsumption on DR days and overconsumption on non-DR days. 
An approximated HighXofY method is proposed to understand how the distribution of the customer consumption and the selection of program parameters affect the calculated baseline. 
The case study shows that the customer baseline manipulation is a single-peaked function in $X$ and remains almost the same under different HighYofY methods. 
The linear policy based rollout algorithm provides adequate accuracy compared to the results from dynamic programming. These results give practical information for DR program providers to design accurate and robust baseline methods.

\bibliographystyle{IEEEtran}
\bibliography{references}

\begin{thebibliography}{10}
\providecommand{\url}[1]{#1}
\csname url@samestyle\endcsname
\providecommand{\newblock}{\relax}
\providecommand{\bibinfo}[2]{#2}
\providecommand{\BIBentrySTDinterwordspacing}{\spaceskip=0pt\relax}
\providecommand{\BIBentryALTinterwordstretchfactor}{4}
\providecommand{\BIBentryALTinterwordspacing}{\spaceskip=\fontdimen2\font plus
\BIBentryALTinterwordstretchfactor\fontdimen3\font minus
  \fontdimen4\font\relax}
\providecommand{\BIBforeignlanguage}[2]{{%
\expandafter\ifx\csname l@#1\endcsname\relax
\typeout{** WARNING: IEEEtran.bst: No hyphenation pattern has been}%
\typeout{** loaded for the language `#1'. Using the pattern for}%
\typeout{** the default language instead.}%
\else
\language=\csname l@#1\endcsname
\fi
#2}}
\providecommand{\BIBdecl}{\relax}
\BIBdecl

\bibitem{vardakas2015survey}
J.~S. Vardakas, N.~Zorba, and C.~V. Verikoukis, ``A survey on demand response
  programs in smart grids: pricing methods and optimization algorithms,''
  \emph{IEEE Communications Surveys \& Tutorials}, vol.~17, no.~1, pp.
  152--178, 2015.

\bibitem{palensky2011demand}
P.~Palensky and D.~Dietrich, ``Demand side management: demand response,
  intelligent energy systems, and smart loads,'' \emph{IEEE Transactions on
  Industrial Informatics}, vol.~7, no.~3, pp. 381--388, 2011.

\bibitem{chao2011demand}
H.-p. Chao, ``Demand response in wholesale electricity markets: the choice of
  customer baseline,'' \emph{Journal of Regulatory Economics}, vol.~39, no.~1,
  pp. 68--88, 2011.

\bibitem{rossetto2018measuring}
N.~Rossetto, ``Measuring the intangible: an overview of the methodologies for
  calculating customer baseline load in {PJM},'' 2018.

\bibitem{enernoc2011the}
{EnerNOC}, ``The demand response baseline,'' Tech. Rep., 2009.

\bibitem{wijaya2014bias}
T.~K. Wijaya, M.~Vasirani, and K.~Aberer, ``When bias matters: An economic
  assessment of demand response baselines for residential customers,''
  \emph{IEEE Transactions on Smart Grid}, vol.~5, no.~4, pp. 1755--1763, 2014.

\bibitem{mohajeryami2017error}
S.~Mohajeryami, M.~Doostan, A.~Asadinejad, and P.~Schwarz, ``Error analysis of
  customer baseline load ({CBL}) calculation methods for residential
  customers,'' \emph{IEEE Transactions on Industry Applications}, vol.~53,
  no.~1, pp. 5--14, 2017.

\bibitem{wolak2007residential}
F.~A. Wolak, ``Residential customer response to real-time pricing: The anaheim
  critical peak pricing experiment,'' 2007.

\bibitem{england2008docket}
``Filing of changes to day-ahead load response program,'' ISO New England
  Docket No. ER08-538-000, Tech. Rep., 2008.

\bibitem{wang2018overconsume}
X.~Wang and W.~Tang, ``To overconsume or underconsume: Baseline manipulation in
  demand response programs,'' in \emph{2018 North American Power Symposium
  (NAPS)}.\hskip 1em plus 0.5em minus 0.4em\relax IEEE, 2018, pp. 1--6.

\bibitem{borenstein2014peak}
S.~Borenstein, ``{Peak-time rebates: money for nothing?}''
  \url{https://www.greentechmedia.com/articles/read/peak-time-rebates-money-for-nothing},
  2014, [Online; accessed 15-March-2018].

\bibitem{vuelvas2015demand}
J.~Vuelvas and F.~Ruiz, ``Demand response: Understanding the rational behavior
  of consumers in a peak time rebate program,'' in \emph{2015 IEEE 2nd
  Colombian Conference on Automatic Control (CCAC)}.\hskip 1em plus 0.5em minus
  0.4em\relax IEEE, 2015, pp. 1--6.

\bibitem{chao2013incentive}
H.-p. Chao and M.~DePillis, ``Incentive effects of paying demand response in
  wholesale electricity markets,'' \emph{Journal of Regulatory Economics},
  vol.~43, no.~3, pp. 265--283, 2013.

\bibitem{ellman2019model}
D.~Ellman and Y.~Xiao, ``Model predictive control-based battery scheduling and
  incentives to manipulate demand response baselines,'' \emph{arXiv preprint
  arXiv:1909.12349}, 2019.

\bibitem{zhou2017eliciting}
D.~P. Zhou, M.~Balandat, M.~A. Dahleh, and C.~J. Tomlin, ``Eliciting private
  user information for residential demand response,'' in \emph{2017 IEEE 56th
  Annual Conference on Decision and Control (CDC)}.\hskip 1em plus 0.5em minus
  0.4em\relax IEEE, 2017, pp. 189--195.

\bibitem{li2017mechanism}
Y.~Li and N.~Li, ``Mechanism design for reliability in demand response with
  uncertainty,'' in \emph{2017 American Control Conference (ACC)}.\hskip 1em
  plus 0.5em minus 0.4em\relax IEEE, 2017, pp. 3400--3405.

\bibitem{vuelvas2018limiting}
J.~Vuelvas, F.~Ruiz, and G.~Gruosso, ``Limiting gaming opportunities on
  incentive-based demand response programs,'' \emph{Applied Energy}, vol. 225,
  pp. 668--681, 2018.

\bibitem{samadi2010optimal}
P.~Samadi, A.-H. Mohsenian-Rad, R.~Schober, V.~W. Wong, and J.~Jatskevich,
  ``Optimal real-time pricing algorithm based on utility maximization for smart
  grid,'' in \emph{2010 First IEEE International Conference on Smart Grid
  Communications}.\hskip 1em plus 0.5em minus 0.4em\relax IEEE, 2010, pp.
  415--420.

\bibitem{simchi2005logic}
D.~Simchi-Levi, X.~Chen, J.~Bramel \emph{et~al.}, ``The logic of logistics,''
  \emph{Algorithms, and Applications for Logistics and Supply Chain Management,
  Theory}, 2005.

\bibitem{topkis1978minimizing}
D.~M. Topkis, ``Minimizing a submodular function on a lattice,''
  \emph{Operations research}, vol.~26, no.~2, pp. 305--321, 1978.

\bibitem{chen1999accurate}
C.-C. Chen and C.~W. Tyler, ``Accurate approximation to the extreme order
  statistics of gaussian samples,'' \emph{Communications in
  Statistics-Simulation and Computation}, vol.~28, no.~1, pp. 177--188, 1999.

\bibitem{powell2007approximate}
W.~B. Powell, \emph{Approximate Dynamic Programming: Solving the curses of
  dimensionality}.\hskip 1em plus 0.5em minus 0.4em\relax John Wiley \& Sons,
  2007, vol. 703.

\bibitem{bertsekas1999rollout}
D.~P. Bertsekas and D.~A. Castanon, ``Rollout algorithms for stochastic
  scheduling problems,'' \emph{Journal of Heuristics}, vol.~5, no.~1, pp.
  89--108, 1999.

\bibitem{bertsekas2013rollout}
D.~P. Bertsekas, ``Rollout algorithms for discrete optimization: A survey,''
  \emph{Handbook of combinatorial optimization}, pp. 2989--3013, 2013.

\bibitem{goodson2017rollout}
J.~C. Goodson, B.~W. Thomas, and J.~W. Ohlmann, ``A rollout algorithm framework
  for heuristic solutions to finite-horizon stochastic dynamic programs,''
  \emph{European Journal of Operational Research}, vol. 258, no.~1, pp.
  216--229, 2017.

\bibitem{bertsekas1996neuro}
D.~P. Bertsekas and J.~N. Tsitsiklis, \emph{Neuro-dynamic programming}.\hskip
  1em plus 0.5em minus 0.4em\relax Athena Scientific Belmont, MA, 1996, vol.~5.

\bibitem{barker2012smart}
S.~Barker, A.~Mishra, D.~Irwin, E.~Cecchet, P.~Shenoy, J.~Albrecht
  \emph{et~al.}, ``Smart*: An open data set and tools for enabling research in
  sustainable homes,'' \emph{SustKDD, August}, vol. 111, no. 112, p. 108, 2012.

\end{thebibliography}

\appendices

\section{The Proof of the Optimal Policies}
\subsection{Proof of Theorem \ref{thm:over_under}}
\begin{proof}
The proof is showed for the DR days and non-DR days. (i) On a DR day ($\yt = 1$), the state update rule (\ref{eq:x_transition}) shows that $\xbmtone = \xbmt$, which means the $\Vtone(\sbmtone)$ is independent of $\at$. Therefore, the corresponding optimal policy can be found via
\begin{equation}
    \label{eq:ats_DR}
    \ats(\xbmt, 1, \zt) \in \argmax_{\atrng} \{{\utcaz + \rbp [\Bma(\xbmt)-\at]_{+}}\}.
\end{equation}
Because the rebate $\rbp [\Bma(\xbmt)-\at]_{+}$ is decreasing in $\at$, one can easily show that $\ats(\xbmt, 1, \zt) \le \atB(\zt)$.
(ii) On a non-DR day ($\yt = 0$), the rebate is zero and the optimal policy can be found through
\begin{equation}
    \label{eq:ats_nonDR}
    \ats(\xbmt, 0, \zt) \in \argmax_{\atrng} \{ {\utcaz + \Ebb_{\sbmtone \vert \sbmt} [\Vtone(\sbmtone)]} \}.
\end{equation}
Suppose the expectation term is increasing in $\at$, one can show that $\atB(\zt) \le \ats(\xbmt, 0, \zt)$. Note that the expectation preserves monotonicity. It is equivalent to show that $\Vtone(\sbmtone)$ is increasing in $x_{t+1, 1}$ ($x_{t+1, 1} = \at$ when $\yt = 0$), which is the result from Lemma \ref{lma:Vt_mono}. 
\end{proof}
\subsection{Proof of Theorem \ref{thm:ats_threshold}}
\begin{proof}
For convenience, we denote the $\utcaz$ as $\utc(\at)$ and $\Bma(\xbmt)$ as $\Bma$. The optimization problem on a DR day (\ref{eq:ats_DR}) can be divided into two subproblems:
\begin{equation}
    \label{eq:ats_opt_0}
    \ats \vert_{\Bma -\at \le 0} \in \argmax_{\Bma \le \at \le \amax } {\utca},
\end{equation}
\begin{equation}
    \label{eq:ats_opt_1}
    \ats \vert_{\Bma-\at \ge 0} \in \argmax_{0 \le \at \le \Bma} \{ {\utca + \rbp(\Bma - \at)} \}.
\end{equation}
The corresponding optimal policies can be obtained as
\begin{equation}
    \label{eq:ats_0}
    \atsn \coloneqq \ats \vert_{\Bma -\at \le 0} = \max\{\atB, \Bma\} \ge \atB,
\end{equation}
\begin{equation}
    \label{eq:ats_1}
    \atsp \coloneqq \ats \vert_{\Bma-\at \ge 0} = \min\{\atU, \Bma\} \le \atU.
\end{equation}
Note that we have $\atU \le \atB$ since $\rbp(\Bma - \at)$ is decreasing in $\at$.
Then, the optimal policy will be chosen by selecting the larger optimal value, given by 
\begin{equation}
    \label{eq:ats_oneInTwo}
    \ats \vert_{\yt=1} \in  \argmax_{\atsn, \atsp} \{\Vt(\sbmt) \vert_{\Bma -\at \le 0}, \Vt(\sbmt) \vert_{\Bma-\at \ge 0}\}.
\end{equation}
The optimal values under these two cases are given by
\begin{equation}
    \label{eq:V_0}
    \Vtn \coloneqq \Vt(\sbmt) \vert_{\Bma -\at \le 0} = \begin{cases}
    \utc(\atB), & \text{if } \Bma \le \atB,\\
    \utc(\Bma), & \text{if } \Bma > \atB.
    \end{cases}
\end{equation}
\begin{equation}
    \label{eq:V_1}
    \Vtp \coloneqq \Vt(\sbmt) \vert_{\Bma-\at \ge 0}
     = \begin{cases}
    \utc(\Bma), & \text{if } \Bma \le \atU, \\
    \utc(\atU) + \rbp (\Bma - \atU), & \text{if } \Bma > \atU.
    \end{cases}
\end{equation}
Note that $\Vtn(\Bma)$ is a fixed value $\utc(\atB)$ in $[0, \atB]$, but decreasing in $[\atB, \amax]$. $\Vtp(\Bma)$ increasing in its range $[0,\amax]$ since $\atU \le \atB$. Furthermore, at $\Bma = \atU$, we have $\Vtn(\atU) = \utc(\atB) > \Vtp(\atU) = \utc(\atU)$. At $\Bma = \atB$, one can prove $\Vtn(\atB) < \Vtp(\atB)$ by showing 
\begin{equation}
    \label{eq:inequality_concave}
    \utc(\atB) - \utc(\atU) < \rbp (\atB - \atU).
\end{equation}
This is true due to the concavity of $\utc(\at)$ and $\dutc(\atB) = 0$, $\dutc(\atU)-r=0$ from the first-order optimality condition. Therefore, there must be a threshold baseline $\Bmath$ that can be obtained via
\begin{equation}
    \label{eq:Bth_threshold_eq}
    \utc(\atB) = \utc(\atU) + \rbp (\Bmath - \atU),
\end{equation}
which is equivalent to (\ref{eq:Bth_thresholdBaseline}). The corresponding optimal policy can be determined by (\ref{eq:ats_oneInTwo}) accordingly, given in (\ref{eq:ats_thresholdPolicy}).
\end{proof}

\subsection{Proof of Theorem \ref{thm:ats_nonDR_supermodular}}
\begin{proof}
For convenience, we denote the $\utcaz$ as $\utc(\at)$ and $\Bma(\xbmt)$ as $\Bma$. The optimal policy on a non-DR day can be reformulated as
\begin{equation}
    \label{eq:ats_nonDR_v2}
    \ats(\xbmt, 0, \zt) \in \argmax_{\atrng} \qt(\xbmt, \at),
\end{equation}
where $\qt(\xbmt, \at)$ is the state function defined as 
\begin{equation}
    \label{eq:q_function}
    \qt(\xbmt, \at) = \utca + \Ebb_{\sbmtone \vert \sbmt}[\Vtone(\sbmtone)].
\end{equation}
Suppose the state function $\qt(\xbmt, \at)$ is supermodular in $(\xbmt, \at)$, one can show $\ats(\xbmt, 0, \zt)$ is increasing in $\xbmt$ based on the Topkis's theorem. It is evident that $\utca$ is supermodular in $(\xbmt, \at)$. We can show $\qt(\xbmt, \at)$ is supermodular if the expectation term $\Ebb_{\sbmtone \vert \sbmt} [\Vtone(\sbmtone)]$ is also supermodular in $(\xbmt, \at)$. Note that when $\yt = 0$, the state on day $t+1$ can be written as
\begin{equation}
    \label{eq:stone_state}
    \sbmtone = (\at, x_{t,1}, \ldots, x_{t,Y-1},\; \ytone, \ztone).
\end{equation}
Therefore, the only missing thing is to show that $\Vtone(\sbmtone)$ is supermodular in $\xbmtone$, which can be proven by induction. 

At stage $t=T-1$, we have $\Vtone(\sbmtone) = V_{T}(\sbm_{T}) = 0$. As a result, $\Vt(\styt) = \max_{\atrng} \{\utca\}$ is independent of $\xbmt$ and thus supermodular. On the other hand,
\begin{equation}
    \label{eq:Vt_DR_T-1}
    \begin{split}
    \Vt(\sbmt \vert \yt = 1) & = \max_{\atrng} \{ {\utca + \rbp[\Bma - \at]_{+}} \} \\
    & = \begin{cases}
    \utc(\atB, \zt), & \text{if } \Bma \le \Bmath,\\
    \utc(\atU, \zt) + \rbp(\Bma - \atU), & \text{if } \Bma > \Bmath.
    \end{cases}
    \end{split}
\end{equation}
According to the result from theorem \ref{thm:ats_threshold}, $\Vt(\Bma\vert \yt = 1)$ is convex and increasing in $\Bma$ . Based on the assumption, $\Bma$ is increasing and supermodular in $\xbmt$. Apply the composition rule from lemma \ref{lma:Vt_mono}, one can show that $\Vt(\xbmt, 1, \zt)$ is supermodular in $\xbmt$. Now suppose the $\Vtone(\sbmtone)$ is supermodular in $\sbmtone$, then we need to show $\Vt(\sbmt)$ is supermodular in $\xbmt$. When $\yt = 1$, the optimal value function is 
\begin{equation}
    \label{eq:Vts_DR}
    \begin{split}
    \Vt(\sbmt \vert \yt = 1) & = \max_{\atrng}\{ {\utca + \rbp[\Bma-\zt]_{+}} \}\\
    & \quad \; + \Ebb_{\sbmtone \vert \sbmt} [\Vtone(\sbmtone)].
    \end{split}
\end{equation}
The first term has been proved to be supermodular in $\xbmt$ above. The second expectation term is supermodular in $\xbmtone$, i.e., $\xbmt$ when $\yt = 1$, by assumption. Then the statement is true for $\yt = 1$. When $\yt = 0$, we have the optimal value function as
\begin{equation}
    \label{eq:Vts_nonDR}
    \Vt(\sbmt \vert \yt = 1) = \max_{\atrng} \{ {\utca + \Ebb_{\sbmtone \vert \sbmt} [\Vtone(\sbmtone)]} \}.
\end{equation}
By assumption, $\Vtone(\sbmtone)$ is supermodular in $\xbmtone = (\at, x_{t,1}, \ldots, x_{t,Y-1})$. Therefore, the whole term inside the $\max\{\cdot\}$ is supermodular in $(\xbmt, \at)$. As a result, $\Vt(\sbmt \vert \yt = 1)$ is supermodular in $\xbmt$ since the maximization over a lattice $\atrng$ preserves the supermodularity \cite{simchi2005logic}. Thus the statement $\Vt(\sbmt)$ is supermodular in $\xbmt$ has been proved, as well as the main conclusion in this theorem.
\end{proof}
\subsection{Proof of Theorem \ref{thm:monotonicity_X}}
\begin{proof}
The approximated HighXofY can be written as a function of $\xti$ and $X$, i.e., $\Haxy(\xti, X) = \xbmtbar + \frac{Y-X}{Y-1}\fbm(Y) \st(\xti)$. The standard deviation can be reformulated as
\begin{equation}
    \label{eq:Haxy_X}
    \st(\xti) = \sqrt{\frac{Y-2}{Y-1} \stysqr + \frac{(\xti - \xbmtybar)^{2}}{Y}},
\end{equation}
where $\xbmtybar$ and $\stysqr$ are the sample mean and sample variance of the rest $Y-1$ data $\{x_{t,j}\}_{j=1,\ldots,Y, j \neq i}$. Specifically, we have the expressions as $\xbmtybar = \frac{1}{Y-1}\mathop{\sum}_{j=1,\ldots,Y, j \neq i} (x_{t,j})$ and $\stysqr = \frac{1}{Y-2}\mathop{\sum}_{j=1,\ldots,Y, j \neq i} {(x_{t,j} - \xbmtybar)^{2}}$.

Clearly we can see that $\st(\xti)$ is decreasing in $(0, \xbmtybar)$ and increasing in $(\xbmtybar, +\infty)$. As a result, one can show that the $\Haxy(\xti, X)$ is supermodular in $(\xti, X)$ when $\xti < \xbmtybar$, and submodular in $(\xti, Y-X)$ when $\xti > \xbmtybar$, according to the definition of supermodularity or its second partial derivative. Then, as long as the intrinsic consumption is less than the average of the other $Y-1$ data, then $\ats(\styt)$ is increasing in $X$. We can roughly treat the sample average the $\xbmtybar$. On the other hand, if the $\at$ is large enough, i.e., larger than the population mean, then $\ats(\styt)$ is decreasing in $X$.
\end{proof}

\end{document}